\documentclass[superscriptaddress,onecolumn,notitlepage,pra]{revtex4-1}
\usepackage[colorlinks=true,citecolor=blue,urlcolor=black]{hyperref}
\usepackage{graphicx}
\usepackage{dcolumn}
\usepackage{bm}
\usepackage{amsthm}
\usepackage{amsmath}
\usepackage{amssymb}
\usepackage{color}
\usepackage{multirow}
\usepackage{algorithm}
\usepackage{algorithmicx}
\usepackage{algpseudocode}
\usepackage{subfigure}

\newtheorem{theorem}{Theorem}

\definecolor{Red}{rgb}{1,0,0}

\def\ket#1{| #1 \rangle}
\def\bra#1{\langle #1 |}

\def\RR{\mathbb{R}}

\def\poly{\operatorname{poly}}

\def\H{\mathcal{H}}

\def\R{\mathbb{R}}

\def\RR{\mathbb{R}}
\def\CC{\mathbb{C}}

\begin{document}

\title{A Query-based Quantum Eigensolver}

\author{Shan Jin}
\affiliation{Institute of  Fundamental and Frontier Sciences, University of Electronic Science and Technology of China, Chengdu, 610051, China}

\author{Shaojun Wu}
\affiliation{Institute of  Fundamental and Frontier Sciences, University of Electronic Science and Technology of China, Chengdu, 610051, China}

\author{Guanyu Zhou}
\affiliation{Institute of  Fundamental and Frontier Sciences, University of Electronic Science and Technology of China, Chengdu, 610051, China}

\author{Ying Li}
\affiliation{Graduate School of China Academy of Engineering Physics, Beijing 100193, China}

\author{Lvzhou Li}
\affiliation{Institute of Computer Science Theory, School of Data and Computer Science, Sun Yat-sen University, Guangzhou, 510006, China}

\author{Bo Li}
\affiliation{Key Laboratory of Mathematics Mechanization, Academy of Mathematics and Systems Science, Chinese Academy of Sciences, Beijing 100190, China}

\author{Xiaoting Wang}
\email{xiaoting@uestc.edu.cn}
\affiliation{Institute of  Fundamental and Frontier Sciences, University of Electronic Science and Technology of China, Chengdu, 610051, China}

\begin{abstract}
Solving eigenvalue problems is crucially important for both classical and quantum applications. Many well-known numerical eigensolvers have been developed, including the QR and the power methods for classical computers, as well as the quantum phase estimation(QPE) method and the variational quantum eigensolver for quantum computers. In this work, we present an alternative type of quantum method that uses fixed-point quantum search to solve Type II eigenvalue problems. It serves as an important complement to the QPE method, which is a Type I eigensolver. We find that the effectiveness of our method depends crucially on the appropriate choice of the initial state to guarantee a sufficiently large overlap with the unknown target eigenstate. We also show that the quantum oracle of our query-based method can be efficiently constructed for efficiently-simulated Hamiltonians, which is crucial for analyzing the total gate complexity. In addition, compared with the QPE method, our query-based method achieves a quadratic speedup in solving Type II problems.
\end{abstract}

\maketitle

\section{Introduction}\label{sec1}

Quantum algorithms for quantum circuits have demonstrated their potential advantages in computational complexity over their classical counterparts, in solving various mathematical problems, such as the integer factorization problem~\cite{shor1994algorithms}, unsorted database search problem~\cite{grover1997quantum}, linear equation problem~\cite{harrow2009quantum} and simulating quantum systems~\cite{lloyd1996universal,berry2007efficient}. Another typical problem is the eigenvalue problem, important both in theory and in applications, and many useful numerical methods have so far been proposed for both classical and quantum circuits. For classical algorithms, well-known eigensolvers include the QR method~\cite{francis1961qr,francis1962qr,kublanovskaya1962some}, the Jacobi method~\cite{carl1846jacobi,golub2000eigenvalue}, and the Sturm sequences method~\cite{gupta1972solution,gupta1973eigenproblem}; for quantum algorithms, two major methods have been developed. The first one is the quantum phase estimation(QPE) \cite{nielsen2002quantum,abrams1999quantum} (a subcircuit of Shor's algorithm \cite{shor1994algorithms}). It uses the quantum Fourier transform(QFT) to find the eigenvalues and the corresponding eigenvectors for a given unitary operator. If the input state of the eigenstate register is in a superposition of different eigenstates, then the output state becomes an entangled state between the eigenvalue and the eigenstate registers. Hence a final measurement after the QPE circuit will derive one of the eigenvalues in the spectrum. The other useful quantum eigensolver is the variational quantum eigensolver(VQE)~\cite{peruzzo2014variational,wang2019accelerated}. VQE uses the quantum-classical hybrid computing architecture. It uses a parametrized quantum circuit to prepare the final state, and then uses classical computer to analyze the measurement results and optimize the parameters. Hence, VQE can be used to first find the ground energy of a given Hamiltonian matrix, then find the next lowest eigenvalue, and then one-by-one find all the eigenvalues in ascending order. In addition, a VQE algorithm on the full quantum system has also been proposed recently~\cite{wei2020full}.

Besides the QPE circuit, the Grover's search algorithm~\cite{grover1997quantum} also has wide applications in quantum algorithm design. It rotates the initial state to the target state through a sequence of Grover iterations~\cite{grover1997quantum,grover1998quantum}. Compared with the classical query complexity $O(N)$, Grover's algorithm achieves a quadratic speedup \cite{nielsen2002quantum,hoyer2000arbitrary} and such quadratic speedup has been proved optimal~\cite{bennett1997strengths,boyer1998tight,zalka1999grover}.
Many applications of quantum search and the improvement of the method have been proposed \cite{brassard1997exact, grover2005fixed, yoder2014fixed,long1999phase, long2001grover,long2002phase, toyama2013quantum}. The application in amplitude amplification was proposed by Brassard et al based on the quantum search \cite{brassard1997exact, brassard2002quantum}, and was independently discovered by Grover in 1998~\cite{grover1998quantum}. However, without knowing how many solutions the search problem has, it is difficult to determine the exact number of Grover iterations, resulting in the possibility of ``undercooking" or ``overcooking"~\cite{long2001grover}. In order to solve this problem, Grover then proposed a fixed-point search method, which guarantees the initial state to converge monotonically to the target state, though the quadratic speedup is lost~\cite{grover2005fixed}. Later, an improved version of the fixed-point search was proposed, regaining the advantage of quadratic speedup~\cite{yoder2014fixed}.

In this work, inspired by the fixed-point search algorithm, we propose a query-based method to solve the Type II eigenvalue problem(i.e., finding the eigenvalue near a given value point). We set the unknown target eigenstate as the search target and transform the Type II eigenvalue problem into a query-based search problem. We show how to use the phase estimation circuit to construct the oracle of the fixed-point quantum search. We also discuss how to choose the initial states to guarantee a sufficiently large overlap with the unknown target eigenstate, which is crucial for the efficiency of the amplitude amplification, and in the meanwhile these initial states are easy to prepare in experiment. Thus, a query complexity of $O(\sqrt{N})$ can be achieved. In particular, when the oracle can be efficiently implemented, i.e., when the gate complexity of the QPE circuit is $O(\text{poly}(\log N))$, the entire gate complexity of our method becomes $O(\sqrt N\text{poly}(\log N))$, demonstrating a quadratic speedup over the complexity of the conventional QPE method to solve the same problem. Finally, as two examples, we apply our method to solving the Type II eigenvalue problems for the Heisenberg model and the hydrogen molecule.


\section{The eigenvalue problem and its classical algorithms}\label{sec2}

The general eigenvalue problem is formulated as follows: for a linear operator $A$ defined on a $N$-dimensional Hilbert space $\H$, we hope to find the eigenvalues $\lambda_k$ and the corresponding eigenvectors $x_k$ such that $Ax_k=\lambda_k x_k$, for all $k$. In the matrix form, $A$ is an $N\times N$ complex matrix. When $A$ is a normal matrix, it has $N$ number of eigenvalues(some may be degenerate) and $N$ independent eigenvectors that form an orthogonal basis for $\H$. The basic assumption for classical eigensolvers is that $A$ has $N$ independent eigenvectors, not necessarily orthogonal to each other. In other words, we only discuss matrices that can be diagonalized through similarity transformation. In addition, we can identify two types of eigenvalue problems:
\begin{itemize}
\item[I.] to find out the whole spectrum;
\item[II.] to find out particular eigenvalues in the spectrum, e.g., the eigenvalue near a given point, or the largest or the second-largest eigenvalue.
\end{itemize}
For Type I problems, the QR method~\cite{francis1961qr,francis1962qr,kublanovskaya1962some} is one of the most popular and efficient algorithms for simultaneously calculating all the eigenvalues and the eigenvectors of $A$, especially when $A$ is not sparse. The Jacobi and the Sturm sequences methods are also typical numerical algorithms to compute the whole spectrum of a Hermitian, sparse matrix. For problems of type II, the power and the inverse power methods are popular algorithms. The power method is good at approximating the eigenvalue with the maximum module, while the inverse power method is utilized to find the eigenvalue with the minimum module or the eigenvalue closest to a given point. In addition, the inverse power method is efficient to determine the eigenvector corresponding to a given eigenvalue \cite{wilkinson1988algebraic}.

In terms of complexity, the computational cost is $O(N^3)$ for the QR, the Jacobi, the Sturm sequences and the inverse power methods, and $O(N^2)$ for the power method~\cite{quarteroni2010numerical}.

\section{The QPE circuit as a quantum eigensolver}\label{sec3}

One of the basic questions in quantum computation is, for the given computational problem, whether there exists a quantum algorithm whose gate complexity is strictly better than the extant classical algorithms. Typical examples are Shor's algorithm and Grover's search algorithm, which obtain exponential and quadratic speedups over their classical counterparts. For many applications of eigenvalues problems, where the dimension of the matrix $A$ could become extremely large, such as in big data and quantum mechanics, the $O(N^3)$ complexity is not good enough. For example, for a quantum system composed of $n$ qubits, the total dimension of the system is $N=2^n$, in which case, the above classical eigensolver algorithms, with complexity $O(N^3)=O(8^n)$ or $O(N^2)=O(4^n)$, become inefficient as $n$ increases. This represents an exponential complexity with respect to $n$. Can we find a quantum eigensolver whose gate complexity is $O(\poly(n))$? This is the question we would like to explore in this work. We will first study the quantum phase estimation(QPE) method, which is a Type I eigensolver.

The idea of using QPE circuit to solve the eigenvalue problem is straightforward. We assume $A$ is Hermitian and has the spectral decomposition: $A=\sum_m\lambda_m\ket{u_m}\bra{u_m}$, with $A\ket{u_m}=\lambda_m\ket{u_m}$. If $A$ is not Hermitian, but is a normal matrix, we can equivalently solve the eigenvalue problem for $A+A^\dag$ and $i(A-A^\dag)$, which are both Hermitian; if $A$ is non-normal, we can discuss a different problem to find the singular value of $A$. In that case, we need to study a new Hermitian matrix $\tilde A$
\begin{align*}
\tilde A=\begin{pmatrix}
0&A\\
A^\dag&0
\end{pmatrix}
\end{align*}
whose spectrum are the singular values of $A$. Alternatively, we can study the eigenvalue problem for $B=A^\dag A$.

In the following, we assume $A$ to be Hermitian. The total quantum processor consists of two parts: the eigenvalue register(containing $r$ qubits) and the eigenvector register(containing $n$ qubits), with $n=  \left \lceil {\log N}\right \rceil $. Without loss of generality, we assume $n=\log N$ to get rid of the cumbersome ceiling notations. Define $U=e^{2\pi iA}$, which is a unitary gate on the eigenvector register. First we initialize the total system as $\ket\Phi=\ket 0\ket\varphi$, where the initial state of the eigenvector register $\ket\varphi$ is randomly chosen from the unit sphere of an $N$-dimensional Hilbert space, with $\ket\varphi=\sum_{i=1}^N b_i\ket {u_i}$ and $b_i=\langle u_i|\varphi\rangle$ as shown in Fig.~\ref{fig1}. We then apply $U_{PE}$ to $\ket\Phi$:
\begin{align}\label{eq1}
U_{PE}\ket0\ket\varphi=\sum_{i=1}^N b_i\ket{\tilde{\lambda}_i}\ket{u_i}
\end{align}
where $\tilde{\lambda}_i$ represents the approximation of the exact value $\lambda_i$ due to the finite length $r$ with error $O(2^{-r})$. Here, we have enclosed $r$ qubits in the eigenvalue register, which determines the precision of the calculated phase. The circuit of quantum phase estimation is shown in Fig.~\ref{fig1}, where $QFT^\dagger$ represents the inverse quantum Fourier transform and $W_H=H^{\otimes r}$, where $H$ is the Hadamard gate on a single qubit. We choose $U^j=e^{2\pi ijA}$, where $j=2^0,2^1,\cdots,2^r$. We are interested in solving the eigenvalue problems for $A$ where $U=e^{2\pi iA}$ can be efficiently simulated on the quantum processor, such as when $A$ is $d$-sparse~\cite{berry2007efficient}.

\begin{figure}
\centering
\includegraphics[totalheight=1.7in,width = 3.5in]{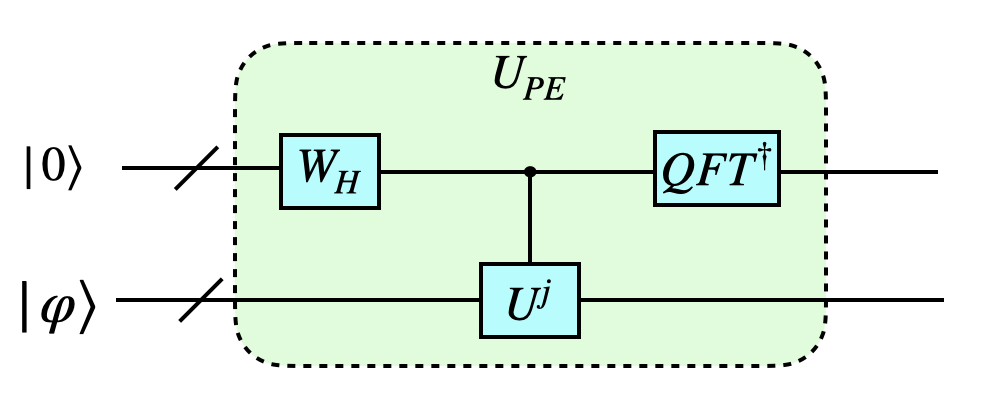}
\caption{QPE circuit to solve a Type I eigenvalue problem. ${QFT}^\dagger $ is the inverse quantum Fourier transform and the unitary transformation $U^j$ is $U^j=e^{2\pi i Aj}$, where $j=2^0,2^1,\cdots,2^r$ and $r$ is the number of qubits in the eigenvalue register, determining the precision of the calculated phase. The unitary transformation $W_{H}$ is $W_H=H^{\otimes r}$, where $H$ is the Hadamard matrix.}
\label{fig1}
\end{figure}

For a given $m$ and a given initial state $\ket\varphi$ randomly chosen from a uniform distribution, we can prove that the probability $P\big(p\equiv |\bra {u_m} \varphi\rangle|^2\ge \frac{1}{N}\big)\ge 1/e$, for sufficiently large $N$, which implies that we can find the minimum $K=11$ such that $1-(1-1/e)^K\ge 0.99$ (with proof details in Appendix~\ref{App1}). That is, if we repeat the random selection of the initial states for $11$ times and get $11$ initial states, $\ket{\varphi_k}$, $k=1,\cdots,11$, then, with probability larger than $0.99$, at least one of $\ket{\varphi_k}$ satisfies $|\bra {u_m} \varphi_k\rangle|^2\ge \frac{1}{N}$. Hence, for each of the $11$ initial states, if we repeat implementing the QPE circuit and taking the final measurement for $O(N)$ times, then we would be able to find all $\lambda_m$ and $\ket{u_m}$. Thus, a Type I problem can be solved by the QPE method, with $O(11N)=O(N)$ number of implementations of QPE circuits, giving a total gate complexity of $O(N\text{poly}(\log N))$, when $U=e^{2\pi iA}$ can be efficiently simulated.

It is worthwhile to mention that if we use the QPE method to solve a Type II problem, then we still need to implement the QPE circuits for $O(11N)=O(N)$ times on average to obtain the desired eigenvalue $\lambda_m$ near the given point $\lambda_0$. One interesting question is whether can find a quantum algorithm to solve the Type II problem with a better complexity, which is the major motivation to study the query-based eigensolver.

%
%

\section{Construction of Query-based Eigensolver}\label{sec4}

\subsection{Fixed-point quantum search}

First, we briefly review how to implement the fixed-point quantum search. The advantage of fixed-point search over the non-fixed-point search is the asymptotic convergence of the system state to the target state after a sequence of fixed-point Grover iterations~\cite{grover2005fixed}. Just like Grover's original search design, the fixed-point search can also achieve the quadratic speedup in query complexity~\cite{yoder2014fixed}. Specifically, in an oracle-based search problem, there is a classical oracle function $f$ on the total set $A$ satisfying $f(x)=1$ if and only if $x \in Q\subset A$, with $|A|=N$ and $|Q|=M\le N$. The goal is to find any of the $M$ solutions in $Q$. In standard Grover's search algorithm, one can construct an $N$-dimensional quantum system to encode the classical data $x\in A$ into its basis vectors $\{\ket x\}$. The initial state is chosen as $\ket \varphi=\frac{1}{\sqrt N}\sum_{x\in A} \ket x$, and the target state as $|t\rangle=\frac{1}{\sqrt M} \sum_{x\in Q}\ket x$. By appending an ancilla qubit $\ket b$, the quantum oracle $O_f$ can be expressed as:
\begin{align*}
O_f|x\rangle|b\rangle=
\begin{cases}
    |x\rangle| b\oplus1\rangle, \quad &x\in Q\\
   |x\rangle|b\rangle, \quad  &x\notin Q
\end{cases}
\end{align*}
Based on $O_f$, we can construct the parametrized inversion operators $R_t(\beta_k)$ and $R_\varphi(\alpha_k)$ with respect to $\ket\varphi$ and $\ket t$:
\begin{align*}
    R_t(\beta_k)&=I-(1-e^{i\beta_k})\ket t\bra t\\
    R_\varphi(\alpha_k)&=I-(1-e^{i\alpha_k})\ket\varphi\bra\varphi
\end{align*}
where $\alpha_k$ and $\beta_k$ are angle parameters. It can be shown that by choosing appropriate values of $\alpha_k$ and $\beta_k$, $k=1,\cdots,l$, the following Grover iteration sequence
\begin{align}
U^{(l)}=G(\alpha_l,\beta_l)\cdots G(\alpha_1,\beta_1)
\end{align}
will make the final state converge to target state $\ket t$ asymptotically as $l$ increases, satisfying $\left|\bra t U^{(l)}\ket\varphi\right|^2\ge1-\delta$, where $G(\alpha_j,\beta_j)=-R_t(\beta_j)R_\varphi(\alpha_j)$ and $l$ is the number of iteration. When $\alpha=\beta=\pi$, $G$ is reduced to the original non-fixed-point Grover iteration; when $\alpha=\beta=\frac{\pi}{3}$ (or $-\frac{\pi}{3}$ at some positions of $R_\varphi$ and $R_t$ in the oracle sequence), $G$ is reduced to Grover's $\pi/3$ fixed-point algorithm; when the value of $\{\alpha_k,\beta_k\}$ is chosen according to the results by Yoder et al~\cite{yoder2014fixed}, $G$ is reduced to Yoder-Luo-Chuang(YLC)'s fixed-point algorithm, with a quadratic speedup. Since each $R_t(\beta_k)$ contains two quantum oracles $O_f$,

Define $\mu=|\bra t \varphi\rangle |^2=\frac{M}{N}$. Assuming that each iteration requires two queries, we find the following optimal value of $q$ which can achieve the error threshold $\delta$:
\begin{align*}
q&=\frac{\ln(\delta/2)}{\ln(1-\mu)}-1&\text{ (Grover's }\pi/3\text{ method)}\\
q &= \frac{\ln(2/\sqrt{\delta})}{\sqrt{\mu}}-1 &\text{ (YLC's method)},
\end{align*}
from which we can see that YLC's method can achieve quadratic speedup: $q=O\big(\frac{1}{\sqrt{\mu}}\big)=O(\sqrt{N})$.

%

\subsection{Solving eigenvalue problems by quantum search}

Next, we show how to use fixed-point Grover's search to solve a Type II eigenvalue problem on a quantum processor. Given an $N\times N$ Hermitian matrix $A$, we assume it has the spectral decomposition: $A=\sum_m\lambda_m\ket{u_m}\bra{u_m}$, with $A\ket{u_m}=\lambda_m\ket{u_m}$, which is unknown to us before the calculation. To solve it on a quantum processor, we consider $A$ as a Hermitian operator on an $n$-qubit system, with $n\equiv  \left \lceil {\log N}\right \rceil $. Without loss of generality, we assume $n=\log N$ to get rid of the cumbersome ceiling notations.

The target Type II problem is to find the eigenvalues near a given point $\lambda_0\in \R$, and the corresponding eigenvectors. Mathematically, the problem can be formulated as searching for all $\lambda_m$ within the $\epsilon$-neighborhood of $\lambda_0$: $B_\epsilon(\lambda_0)\equiv [\lambda_0-\epsilon,\lambda_0+\epsilon]$. If more than one solution is located in $B_\epsilon(\lambda_0)$, our method will randomly output one of the solutions at the final measurement; if there is no solution in $B_\epsilon(\lambda_0)$, we can always tell this from the final measurement outcome. When the latter happens, we can enlarge the value of $\epsilon$ and redo the whole process for multiple times until we find a sufficiently large $\epsilon$ such that at least one solution is located in $B_\epsilon(\lambda_0)$. Hence, without loss of generality, in the following we assume there is one and only one eigenvalue located in $B_\epsilon(\lambda_0)$. We further assume it is $\lambda_1$. Then for quantum search, we can define the target state
\begin{align*}
\ket{t}=\frac{1}{\sqrt{M}}\sum_{\lambda_m\in B_\epsilon(\lambda_0)}\ket{u_m},
\end{align*}
where $M$ is the number of eigenvalues within $B_\epsilon(\lambda_0)$. Under the unique-solution assumption, $\ket{t}=\ket{u_1}$.

The basic idea of solving a Type II eigenvalue problem through the quantum search is as follows. First we choose the initial state of the $n$-qubit quantum processor as a superposition of all eigenvectors of $A$: $\ket {\varphi}=\sum_k b_k\ket{u_k}$, such that the overlap $p\equiv |\bra t \varphi\rangle|^2=|b_1|^2$ is nonzero (later we will show we can further guarantee that such $\ket {\varphi}$ with $p\ge \frac{1}{N}$, can be found through a fixed number of random selections). Then quantum search can be applied to amplify the overlap $p$ through a sequence of fixed-point Grover iterations until $p$ is sufficiently close to $1$. Finally a single measurement is enough to output the target eigenvector $\ket{t}=\ket{u_1}$ and the corresponding eigenvalue $\lambda_1$.

Specifically, we need to figure out how to design the quantum oracle $R_t(\beta)$ for this particular problem. The classical oracle $f$ should satisfy $f(u_m)=1$ for $\lambda_m\in B_\epsilon(\lambda_0)$, and $f(u_m)=0$ for $\lambda_m\notin B_\epsilon(\lambda_0)$.
Hence, $R_t(\beta)$ should satisfy
\begin{align*}
R_t(\beta)\ket {u_m}=\begin{cases}
    e^{i\beta}\ket {u_m}, & \mbox{if }\lambda_m\in B_\epsilon(\lambda_0)\\
   \ket {u_m}, & \mbox{if }\lambda_m\notin B_\epsilon(\lambda_0)
    \end{cases}
\end{align*}
The action of $R_t(\beta)$ is to add a relative phase $e^{i\beta}$ to all target solutions $\ket{u_m}$. We need a basic component in $R_t(\beta)$ whose input is an eigenvector $\ket{u_m}$, and whose output is the corresponding eigenvector $\lambda_m$, which determines whether to add the relative phase $e^{i\beta}$. A candidate gate to implement such component is a quantum phase estimation gate $U_{PE}=e^{2\pi i A}$. Indeed, $R_t(\beta)$ can be implemented by the circuit shown in Fig.~\ref{fig2}, where $Z(\beta)=\text{diag}(1,e^{i\beta})$.


\begin{figure}
\centering
\includegraphics[totalheight=2in,width = 5.3in]{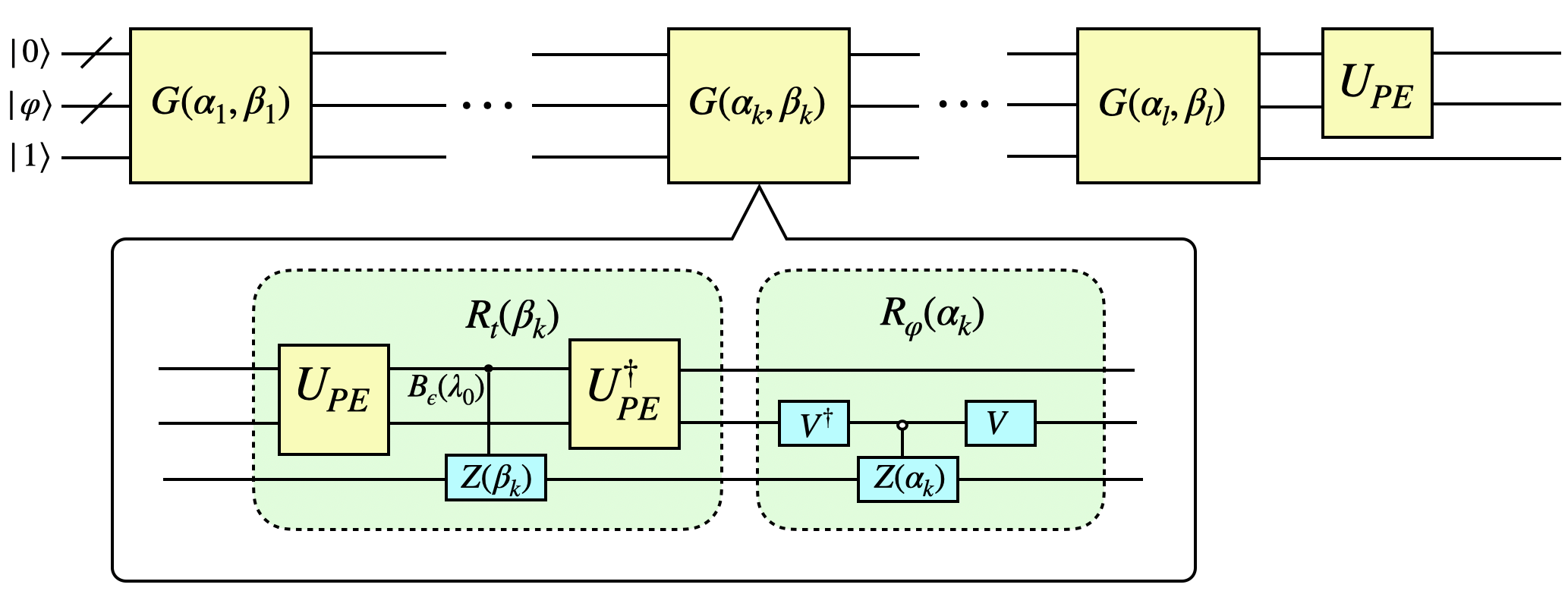}
\caption{Query-based method to solve Type II problems. $U_{PE}$ is the phase estimate mentioned in the previous section.  The Z-gate in the circuit is $Z(\beta)=\text{diag}(1,e^{i\beta})$. The first control gate is  $CZ(\beta)=\sum_{\tilde{\lambda}_j\in B_\epsilon(\lambda_0)} \ket {\tilde{\lambda}_j} \bra {\tilde{\lambda}_j}\otimes Z(\beta) + \sum_{\tilde{\lambda}_j\notin B_\epsilon(\lambda_0)} \ket {\tilde{\lambda}_j} \bra {\tilde{\lambda}_j}\otimes I$. The second control gate is $CZ'(\alpha)=|0\rangle \langle 0|\otimes Z(\alpha)+|1\rangle \langle 1|\otimes I$. $R_t(\beta)$ is the rotation of the target state and $R_\varphi(\alpha)$  is the rotation of the initial state.}
\label{fig2}
\end{figure}

Hence, the total quantum processor is composed of three registers, the $r$-qubit eigenvalue register initialized as $\ket 0$, the $n$-qubit eigenvector register initialized as $\ket\varphi$, and the ancilla single-qubit register initialized as $\ket 1$. Then the total initial state becomes $\ket\Phi=\ket 0\ket\varphi \ket 1$ as in Fig.~\ref{fig2}. In addition, due to the finite precision of the eigenvalue register, each eigenvalue $\lambda_m$ is denoted by a finite-precision state $\ket{\tilde\lambda_m}$. Define the conditional phase gate $CZ(\beta)$ on the eigenvalue register and the ancilla qubit to be:
\begin{align}
CZ(\beta)=\sum_{\tilde{\lambda}_m\in B_\epsilon(\lambda_0)} \ket {\tilde{\lambda}_m} \bra {\tilde{\lambda}_m}\otimes Z(\beta) + \sum_{\tilde{\lambda}_m\notin B_\epsilon(\lambda_0)} \ket {\tilde{\lambda}_m} \bra {\tilde{\lambda}_m}\otimes I
\end{align}
Then the action of $R_t(\beta)\equiv U^\dagger_{PE}\cdot CZ(\beta)\cdot U_{PE}$ on $\ket\Phi$ gives:
\begin{align*}
\begin{split}
R_t(\beta)\ket\Phi=&U^\dagger_{PE}CZ(\beta)U_{PE}|0\rangle|\varphi\rangle|1\rangle\\
\approx&U^\dagger_{PE}CZ(\beta)\sum_{i=1}^{N} b_i |\tilde{\lambda}_i\rangle|u_i\rangle\otimes|1\rangle\\
=&U^\dagger_{PE}\big(\sum_{i=2}^{N} b_i |\tilde{\lambda}_i\rangle|u_i\rangle+b_1e^{i\beta}|\tilde{\lambda}_1\rangle|u_1\rangle\big) \otimes|1\rangle\\
=&\ket 0 \big(\sum_{i=2}^{N} b_i |u_i\rangle+e^{i\beta}b_1|u_1\rangle\big)\ket 1,
\end{split}
\end{align*}
where the approximate sign $\approx$ is due to the finite precision of $\tilde{\lambda}_m$ to represent $\lambda_m$. We can see that such construction of $R_t(\beta)$ is exactly what we need to implement the YLC's fixed-point search method~\cite{yoder2014fixed}. On the other hand, given the initial state $\varphi$, there exists a unitary $V$ such that $\ket\varphi=V\ket 0$. Then we can construct $R_\varphi(\alpha)\equiv V CZ'(\alpha_k)V^\dag$, where $CZ'(\alpha)$ is the controlled-phase gate acting on the eigenvector qubit and the ancilla, and has the form:
\begin{align}
CZ'(\alpha)=|0\rangle \langle 0|\otimes Z(\alpha)+|1\rangle \langle 1|\otimes I
\end{align}

Based on $R_\varphi(\alpha)$ and $R_t(\beta)$, the fixed-point Grover iteration becomes $G=-R_\varphi(\alpha)R_t(\beta)$. Then we choose $\alpha_k$ and $\beta_k$ for each Grover iteration $G(\alpha_k,\beta_k)$ in the search sequence, according to the following rule~\cite{yoder2014fixed}:
\begin{align}
\alpha_j=\beta_{l-j+1}=-2\cot^{-1}(\tan(2\pi j/L)\sqrt{1-\eta^2})
\end{align}
where $\eta^{-1}=T_{1/L}(1/\sqrt{\delta})$, and $l$ is the number of iterations with $L=2l+1$. Here $T_L(x)=\cos[L\cos^{-1}(x)]$ is the $L^{th}$ Chebyshev polynomial of the first kind, and $\delta$ is the error requirement of the search results. So for error requirement $\delta$, we can obtain the optimal query number $q$:
\begin{align}\label{eq2}
q = \frac{\log(2/\sqrt{\delta})}{\sqrt{p}}-1.
\end{align}
where $q=2l$ represents the number of phase-estimation gates used in our proposal~\cite{yoder2014fixed}. If we assume, after $q$ number of fixed-point queries, the eigenvalue and the eigenvector registers becomes:
\begin{align}
\ket\Psi\equiv \ket 0 \ket\psi=\ket 0(c_1|u_1\rangle+\sum_{i\neq 1}c_i |u_i\rangle),
\end{align}
where the overlap $|c_1|^2$ between $\ket{u_1}$ and $\ket{\psi}$ is close to $1$, then applying one more QPE gate to $\ket\Psi$ will generate the final state
\begin{align}\label{eq:1}
U_{PE}|0\rangle|\psi\rangle=c_1|\tilde{\lambda}_1\rangle|u_1\rangle+\sum_{i\neq 1}c_i |\tilde{\lambda}_i\rangle |u_i\rangle,
\end{align}
where $\tilde{\lambda}_i$ is the approximation of the eigenvalue $\lambda_i$ corresponding to the eigenstate $\ket{u_i}$. Hence, when we measure the first register, we can get the desired $\ket{\tilde\lambda_1}$, $\ket{u_1}$ with high probability. Thus, the Type II eigenvalue problem is solved. The algorithm flow chart is summarized in the following.

\begin{algorithm}[H]
\caption{A query-based quantum eigensolver}
\begin{algorithmic}
\Procedure{eigensolver}{$\alpha_{j},\beta_{j},|\varphi_k\rangle,q$}
\State Prepare 11 initial states $|\varphi_k\rangle$
\For {k=1:11}
\State $|\phi^{(0)}\rangle=|\varphi_k\rangle$
\For {j=1:$q$}
\State $R_{\varphi_k}(\alpha_{j})=I-(1-e^{i\alpha_{j}})|\varphi_k\rangle\langle \varphi_k|$
\State $R_{t}(\beta_{j})=I-(1-e^{i\beta_{j}})|t\rangle\langle t|$
\State $|\phi^{(j)}\rangle=
-R_{\varphi_k}(\alpha_{j})R_{t}(\beta_{j})|\phi^{(j-1)}\rangle$
\EndFor
\State $|\psi_k\rangle=|\phi^{(j)}\rangle$
\EndFor
\State \Return{$\{|\psi_k\rangle\}$}
\EndProcedure
\end{algorithmic}
\end{algorithm}

\subsection{Initial state preparation for efficient search}\label{sec4.3}

From Eqn.~(\ref{eq2}), we can see that the greater the overlap between the initial state $\ket{\varphi}$ and target state $\ket{t}$, the fewer number of queries $q$ is required. Hence, the choice of the initial state is crucial in order to keep $q$ small and keep the query-based method efficient. From the discussion in Appendix~\ref{App1}, if the initial state $\ket{\varphi}$ is randomly chosen from the uniform distribution on the unit sphere of $\CC^N$, with $\ket{\varphi}=\sum_{j=1}^{N}b_j \ket{u_j}$ and assuming the target state $\ket{t}=\ket{u_1}$, then the probability of the overlap $p=|b_1|^2=|\bra {u_1}\varphi\rangle|^2$ to be larger than $\frac{1}{N}$ is given by
\begin{align}
P(|b_1|^2\ge\frac{1}{N})=(1-\frac{1}{N})^{N-1}\approx \frac{1}{e},\,\text{for sufficiently large }N.
\end{align}
This implies that we can find the minimum $K=11$ such that $1-(1-1/e)^K\ge 0.99$. In other words, if we prepare a set of $11$ randomly-chosen initial states $\{\ket{\varphi_k}\}$, then the probability of at least one of $\ket{\varphi_k}$'s satisfying $|\bra {\varphi_k}t\rangle | \ge 1/\sqrt{N}$ is larger than $0.99$. Such minimum value $K=11$ is independent of the size $N$ of the eigenvalue problem. For each $\ket{\varphi_k}$, if we implement the fixed-point query sequence with $q = \sqrt{N}\log(2/\sqrt{\delta})-1$ for $11$ times, then at least one of sequence after the final measurement will generate the final state $\ket{\tilde\lambda_1}\ket{u_1}$ with probability larger than $0.99$. Hence, the total query complexity of our method is $O(11q)=O(\sqrt N)$.

In addition, as discussed in the appendix, the set of $11$ randomly-chosen initial states can be substituted by a set of $11$ computational basis states, which can be efficiently generated in the experiment. For example, if we can efficiently generate the ground state $\ket{00\cdots 0}$ of an $n$-qubit system, then all the other computational basis state can be efficiently prepared from the ground state with no more than $\log N=n$ number of bit-flips.

\subsection{Complexity analysis}

Analyzing the complexity of the quantum circuit in Fig.~\ref{fig2}, we can see that the efficiency of the proposed query-based eigensolver algorithm depends on whether the unitary gate $U_{PE}$(or $U=e^{2\pi iA}$ in particular) can be efficiently generated. We can identify two cases when this is valid: (1) if $A$ is a $d$-sparse matrix~\cite{berry2007efficient}; (2) if $A$ is an inherent Hamiltonian that can be generated directly and efficiently on the quantum processor. In these cases, the complexity of generating $U_{PE}$ becomes $O(\poly(\log N))$. As we can prepare the initial state $\ket \varphi$ such that $p=|\bra {u_1}\varphi\rangle|^2\ge \frac{1}{N}$, the query complexity of our method is $q=O(\sqrt{N})$. Since the total gate complexity is equal to the total query complexity multiplied by the oracle gate complexity, i.e. the complexity of $U_{PE}$, we have the total complexity of our algorithm to be $O(\sqrt{N}\text{poly}( \log N))$. Compared with the complexity of using QPE eigensolver to solve a Type II problem, our query-based eigensolver obtains a quadratic speedup.

\section{Applications}\label{sec5}

In this section, we apply our algorithm to two eigenvalues problems in quantum physics. One is the Heisenberg model, and the other is the hydrogen molecule.

\subsection{The Heisenberg model}

The Heisenberg model is a well-known and useful testbed to study many-body physics and quantum information. The Hamiltonian of an $n$-qubit Heisenberg model is
\begin{align}
H = \sum_{j=1}^n J_x\sigma_j^x\sigma_{j+1}^x+J_y\sigma_j^y\sigma_{j+1}^y + J_z\sigma_j^z\sigma_{j+1}^z+h\sigma_j^z,
\end{align}
where $J_x, J_y, J_z$ are coupling constants and $h$ indicates the external magnetic field. $\sigma^k$, $k=x,y,z$ are Pauli operators, and $\sigma_j^k$ denotes $\sigma^k$ acting on the $j$-th qubit. With further assumption of the periodic boundary condition, we assume $\sigma_{n+1}^k=\sigma_1^k$. Given $H$ and $\lambda_0$, our task is to find the eigenvalue $\lambda_1$ of $H$ near $\lambda_0$ and the corresponding eigenstate $\ket{u_1}$. In order to facilitate the simulation, we reconstruct a new Hamiltonian $\tilde{H}=H-\lambda_0 I$, and then the original task is equivalent to solving the Type II eigenvalue problem for $\tilde{H}$ near the point $\tilde\lambda_0=0$. As mentioned earlier in the paper, we can use the query-based quantum eigensolver to solve this problem, using the quantum circuit shown in Fig.~\ref{fig2}.

In the following, we present the simulation results for two cases of the Heisenberg model with $n=4,5$. For each case, we randomly choose the parameters $J_x, J_y, J_z, h$ from the uniform distribution on $[0,1]$. In order to verify the efficiency of our proposed initial-state preparation method, we apply two methods to generate the initial state $\ket{\psi}$: (1) randomly choosing $11$ computational basis states $\ket{x_k}$, $x_k\in\{1,\cdots,N\}$, (2) randomly generating $11$ states $\ket{y_k}$ from the uniform distribution on the unit sphere in $\CC^{n}$. According to the previous discussion, there exists at least one out of the 11 initial states satisfying the overlap $p=b_1\ge\frac{1}{2^n}$ with a high probability($\ge0.99$), assuming the target state is $\ket{u_1}$, and $|b_1|^2=|\bra u_1\varphi\rangle|^2$. Substituting $p=\frac{1}{2^n}$ and error threshold $\delta=0.01$ into Eqn.~(\ref{eq2}), we obtain the minimum values of the query number $q_{\min}=11$ for $n= 4$, and $q_{\min}=16$ for $n=5$. In the circuit design, for each case of $n=4$ and $n=5$, our query-based eigensolver circuit encloses $4$ or $5$ qubits as the eigenvector register, $7$ qubits as the phase register and one more qubit as the ancilla. Next we implement our the query-based eigensolver $11$ times for each initial state $\ket{\psi}=\ket{x_k}$, or $\ket{\psi}=\ket{y_k}$, $k=1,\cdots, 11$.

Table \ref{tab:1} and \ref{tab:2} summarize the simulation results for $n=4$ and $n=5$. For $n=4$, after $q=11$ fixed-point Grover's queries, there are four $\ket{x_k}$'s satisfying $p_1\ge1/16$ with final fidelity $F\ge0.9933$, and three $\ket{y_k}$'s satisfying $p_2\ge1/16$ with $F\ge0.9892$(Table \ref{tab:1}). Analogously, for $n=5$, after $q=16$ fixed-point Grover's queries, there are four $\ket{x_k}$ satisfying $p_1\ge1/32$, with $F_1\ge0.9937$, and six $\ket{y_k}$ satisfying $p_2\ge1/32$, with $F_2\ge0.9724$(Table \ref{tab:2}). Thus, our query-based method is able to solve the Type II eigenvalue problem for the Heisenberg model, within the expected accuracy and computational complexity.

\begin{table}
\begin{tabular}{|c|c|c|c|c|c|c|c|c|c|c|c|} 
\hline
Initial State &$\ket{x_1}$&$\ket{x_2}$&$\ket{x_3}$&$\ket{x_4}$&$\ket{x_5}$&$\ket{x_6}$&$\ket{x_7}$ &$\ket{x_8}$ &$\ket{x_9}$&$\ket{x_{10}}$&$\ket{x_{11}}$\\
\hline
Overlap $p_1$&0.0084&2.416e-33&4.577e-33&0.0645&4.026e-33&0.1404&0.1301&1.490e-32&4.600e-33&0.0645 &1.233e-32\\
\hline
Fidelity $F_1$&0.3931&2.399e-33&4.502e-33&0.9956&4.005e-33&0.9960&0.9933&1.484e-32&4.524e-33&0.9956 &1.225e-32\\
\hline\hline
Initial State&$\ket{y_1}$&$\ket{y_2}$&$\ket{y_3}$&$\ket{y_4}$&$\ket{y_5}$&$\ket{y_6}$&$\ket{y_7}$ &$\ket{y_8}$ &$\ket{y_9}$&$\ket{y_{10}}$&$\ket{y_{11}}$\\
\hline
Overlap $p_2$&0.0545&0.0242&0.0048&0.05632&0.0385&0.1068&0.0329&0.0651&0.1123&0.0451&0.0226\\
\hline
Fidelity $F_2$&0.9921&0.7965&0.2469&0.9944&0.9428&0.9892&0.9016&0.9986&0.9893&0.9726&0.7706\\
\hline
\end{tabular}
\centering
\caption{\label{tab:1}Simulation results for the $4$-qubit Heisenberg model. By randomly selection, we choose $J_x=0.2365, J_y=0.8237, J_z=0.3689$ and $h=0.7326$. We choose $11$ $\ket{x_k}$'s and $11$ $\ket{y_k}$'s as the initial states. Here $p$ denotes the overlap between the initial state and the target state $\ket{t}$, and $F$ denotes the final fidelity after $q=11$ fixed-point Grover's queries. Simulation results show that there are four $\ket{x_k}$'s satisfying $p\ge1/16$ with $F\ge0.9933$, and three $\ket{y_k}$'s satisfying $p\ge1/16$ with $F\ge0.9892$.}
\end{table}

\begin{table}
\begin{tabular}{|c|c|c|c|c|c|c|c|c|c|c|c|} 
\hline
Initial State&$\ket{x_1}$&$\ket{x_2}$&$\ket{x_3}$&$\ket{x_4}$&$\ket{x_5}$&$\ket{x_6}$&$\ket{x_7}$ &$\ket{x_8}$ &$\ket{x_9}$&$\ket{x_{10}}$&$\ket{x_{11}}$\\
\hline
Overlap $p_1$&0.1853&7.704e-34&3.131e-32&0.1105&1.083e-34&0.1128&5.566e-32&1.738e-32&0.1128&0.0285 &6.068e-32\\
\hline
Fidelity $F_1$&0.9937&3.150e-34&2.845e-32&0.9967&4.458e-35&0.9981&4.937e-32&1.804e-32&0.9981&0.9793 &4.241e-32\\
\hline\hline
Initial State&$\ket{y_1}$&$\ket{y_2}$&$\ket{y_3}$&$\ket{y_4}$&$\ket{y_5}$&$\ket{y_6}$&$\ket{y_7}$ &$\ket{y_8}$ &$\ket{y_9}$&$\ket{y_{10}}$&$\ket{y_{11}}$\\
\hline
Overlap $p_2$&0.0275&0.0757&0.0311&0.1294&0.0525&0.0861&0.0074&0.0396&0.1442&0.0235&0.0234\\
\hline
Fidelity $F_2$&0.9261&0.9724&0.9698&0.9838&0.9847&0.9846&0.4863&0.9850&0.9884&0.9218&0.8968\\
\hline
\end{tabular}
\centering
\caption{\label{tab:2}Simulation results for the $5$-qubit Heisenberg model. By randomly selection, we choose $J_x=0.9489, J_y=0.3456, J_z=0.5629$ and $h=0.7475$. Analogous to the $n=4$ case, we choose $11$ $\ket{x_k}$'s and $11$ $\ket{y_k}$'s as the initial states. After $q=16$ fixed-point Grover's queries, there are four $\ket{x_k}$ satisfying $p\ge1/32$ with $F\ge0.9937$, and six $\ket{y_k}$ satisfying $p\ge1/32$ with $F\ge0.9724$.}
\end{table}

\subsection{The hydrogen molecule}

To illustrate our proposed algorithm, we simulate the hydrogen molecule to find an eigenvalue near a given point and its corresponding eigenstate. After Born-Oppenheimer approximation, the Hamiltonian of the hydrogen molecule is as follows:
\begin{align}
H=\sum_{i,j}h_{ij}a^\dagger_i a_j+\frac{1}{2}\sum_{i,j,k,l}h_{ijkl}a^\dagger_i a^\dagger_j a_k a_l,
\end{align}
where the coefficients $h_{ij}$ and $h_{ijkl}$ are one- and two-electron overlap integrals \cite{aspuru2005simulated,whitfield2011simulation,mcweeny1992methods}. The operators $a^\dagger_j$ and $a_j$ are the creation operator and the annihilation operators, respectively. After the Jordan-Wigner transformation \cite{whitfield2011simulation}, the Hamiltonian $H$ becomes
\begin{align}
\begin{split}
H_{JW}=&-0.81261I+0.171201\sigma^z_0+0.171201\sigma^z_1-0.2227965\sigma^z_2-0.2227965\sigma^z_3\\
&+0.16862325\sigma^z_1\sigma^z_0+0.12054625\sigma^z_2\sigma^z_0+0.165868\sigma^z_2\sigma^z_1+0.165868\sigma^z_3\sigma^z_0\\
&+0.12054625\sigma^z_3\sigma^z_1+0.17434925\sigma^z_3\sigma^z_2-0.04532175\sigma^x_3\sigma^x_2
\sigma^y_1\sigma^y_0\\
&+0.04532175\sigma^x_3\sigma^y_2\sigma^y_1\sigma^x_0+0.04532175\sigma^y_3\sigma^x_2\sigma^x_1\sigma^y_0-0.04532175\sigma^y_3\sigma^y_2\sigma^x_1\sigma^x_0
\end{split}
\end{align}

We aim to solve the eigenvalue problem for $H_{JW}$ near a given point $\lambda_0=-0.8837$. For convenience, we reconstruct a new Hamiltonian $\tilde{H}=H_{JW}-\lambda_0 I$. In order to apply our method to $\tilde H$, we enclose $4$ qubits in the eigenvector register, $7$ qubits in the phase register and one more qubit as the ancilla. With the error threshold chosen as $\delta=0.01$, we find the minimum query number required is $q_{\min}=11$. Analogous to the Heisenberg model, we randomly choose $11$ $\ket{x_k}$'s and $11$ $\ket{y_k}$'s as the initial states, which guarantees that at least one of them satisfies $p\ge \frac{1}{2^n}=1/16$. Indeed, simulation results demonstrate that, after $q=11$ queries, there are two $\ket{x_k}$'s satisfying $p\ge1/16$, with $F\ge0.9917$, and there are four $\ket{y_k}$'s satisfying $p\ge1/16$ with $F\ge0.9870$ (as shown in Table~\ref{tab:3}), in line with our expectation.

\begin{table}
\begin{tabular}{|c|c|c|c|c|c|c|c|c|c|c|c|} 
\hline
Initial State&$\ket{x_1}$&$\ket{x_2}$&$\ket{x_3}$&$\ket{x_4}$&$\ket{x_5}$&$\ket{x_6}$&$\ket{x_7}$ &$\ket{x_8}$ &$\ket{x_9}$&$\ket{x_{10}}$&$\ket{x_{11}}$\\
\hline
Overlap $p$&0&0&0&0&0&0&0.5&0&0&0.5&0\\
\hline
Fidelity $F$&0&0&0&0&0&0&0.9917&0&0&0.9917&0\\
\hline\hline
Initial State&$\ket{y_1}$&$\ket{y_2}$&$\ket{y_3}$&$\ket{y_4}$&$\ket{y_5}$&$\ket{y_6}$&$\ket{y_7}$ &$\ket{y_8}$ &$\ket{y_9}$&$\ket{y_{10}}$&$\ket{y_{11}}$\\
\hline
Overlap $p$&0.0211&0.0044&0.0324&0.1028&0.0021&0.0447&0.0790&0.0570&0.1089&0.0116&0.0695\\
\hline
Fidelity $F$&0.7433&0.2238&0.8929&0.9870&0.1146&0.9708&0.9950&0.9947&0.9883&0.5041&0.9973\\
\hline
\end{tabular}
\centering
\caption{\label{tab:3}Simulation results for the hydrogen molecule. Analogous to the Heisenberg model, we choose $11$ $\ket{x_k}$'s and $11$ $\ket{y_k}$'s as the initial states. After $q=11$ queries, there are two $\ket{x_k}$'s satisfying $p\ge1/16$,with $F\ge0.9917$; and there are four $\ket{y_k}$ satisfying $p\ge1/16$ with $F\ge0.9870$.}
\end{table}

\section{Conclusion}\label{sec6}
In this work, we have shown how to use phase estimation circuit to construct the quantum oracle of the fixed-point quantum search to solve Type II eigenvalue problems. This method can serve as a useful complement to the well-known QPE method, which is a typical Type I quantum eigensolver. We have analyzed the gate complexity of our query-based eigensolver and find that, if $U=e^{2\pi iA}$ can be efficiently generated on the quantum processor, then the overall complexity of our method is $O(\sqrt N\, \text{poly}(\log N))$. Compared with the complexity of using the QPE method to solve the same Type I problem, our proposed method has a quadratic speedup. In other words, the QPE method and the query-based method have their respective advantages: one is better at solving Type I problems, and the other is better at Type II. Notice that, in spite of the potential quantum speedup, there is still a gap between the application range of the quantum eigensolvers and that of the classical eigensolvers. Both the QPE method and the query-based eigensolver can only solve for a normal matrix, i.e., a matrix that is unitarily diagonalizable; in comparison, classical eigensolvers, such as the QR method and the power method, can solve for any diagonalizable matrices. How to construct the quantum eigensolver for a diagonalizable but not unitarily-diagonalizable matrix deserves further investigation. Moreover, the question whether one is able to find a quantum eigensolver with a gate complexity $O(\text{poly}(\log N))$ remains open.

\section*{Acknowledgments}
The authors gratefully acknowledge the grant from National Key R\&D Program of China, Grant No.2018YFA0306703. Y. L. thanks National Natural Science Foundation of China, Grant No. 11875050 and NSAF, Grant No. U1930403. L.L. thanks National Natural Science Foundation of China, Grant No. 61772565 and Guangdong Basic and Applied Basic Research Foundation Grant No. 2020B1515020050. B.L. is partly supported by the National Natural Science Foundation of China, Grant No. 61873262. We also thank Xiaokai Hou, Dingding Wen, Yuhan Huang, and Qingyu Li for helpful and inspiring discussions.

\bibliography{ref}

\appendix
\section{Initial-State Preparation}\label{App1}

\begin{theorem}

Let $\ket{\varphi}=\sum_{j=1}^{N}b_j \ket{u_j} $ be the initial state of an $N$-dimensional quantum system, and we assume it is randomly chosen from the unit sphere of $\CC^N$. Without loss of generality, we further assume the target state $\ket t=\ket {u_1}$.
Then the probability of $|\bra{\varphi}t\rangle|^2= |b_1|^2 \ge \frac{1}{N}$ is given by
\begin{equation}\label{eq:Pr-C}
\mathrm{Pr} \left(|b_1|^2 \ge \frac{1}{N}\right) =  \left( 1-\frac{1}{N} \right)^{N-1}.
\end{equation}

\end{theorem}

\begin{proof}

Let $b_j = x_j + y_j i$ with $x_j, y_j \in \mathbb{R}$,
and denote a vector in $\RR^{2N}$ as $z = (x_1,y_1,\cdots,x_N,y_N)$.
Denote $S^{2N-1}_r = \{ z \in \mathbb{R}^{2N} \mid \sum_{j=1}^N (x_j^2 + y_j^2) = r^2 \}$ as the $(2N-1)$ dimensional sphere with radius $r$ in $\RR^{2N}$,
and $D^{2N-1}_r = \{z \in S^{2N-1}_r \mid x_1^2 + y_1^2 \le \frac{r^2}{\xi} \}$ as a closed subset of $S^{2N-1}_r$, with $\xi\ge 1$.
Since
\[
 \sum_{j=1}^N |b_j|^2 = \sum_{j=1}^N (x_j^2 + y_j^2) = 1,
\]
we would like to calculate the following probability:
\[
\mathrm{Pr} \left(|b_1| \le \frac{1}{\sqrt{\xi}}\right) = \mathrm{Pr} \left( x_1^2 + y_1^2 \le \frac{1}{\xi}\right) = \frac{|D^{2N-1}_1|}{ |S^{2N-1}_1| },
\]

The area of $D _r^{N-1}$ is calculated as follows (where $1_\omega$ denotes the characteristic function of the set $\omega$):
\[
\begin{aligned}
& \quad |D_r^{2N-1}|  = \int_{D_r^{2N-1}} 1~dS_r^{2N-1} = \int_{S_r^{2N-1}} 1_{D _r^{2N-1}}(z) ~dS_r^{2N-1} = \frac{\partial}{ \partial r} \int_{V_r^{2N}} 1_{D_{|z|}^{2N-1}}(z) ~dz \\
& = \frac{\partial}{ \partial r}\int_{|z|^2 \le r^2} 1_{D_{|z|}^{2N-1}}(x_1,y_1,\tilde{z})~dx_1 dy_1d\tilde{z} \qquad (\tilde{z} = (x_2,y_2,\cdots,x_N,y_N))\\
& = |S_1^{2N-3} |\frac{\partial}{ \partial r} \int_{|z|^2 =x_1^2 + y_1^2 + \rho^2 \le r^2} 1_{D_{|z|}^{2N-1}} \rho^{2N-3} ~dx_1dy_1 d \rho  \\
& =  |S_1^{2N-3}|\frac{\partial}{ \partial r} \left( \int_0^r r_1~dr_1 \int_0^{2\pi}~d\theta_1   \int_0^{\sqrt{r^2 - r_1^2}} 1_{D_{|z|}^{2N-1}} \rho^{2N-3} d\rho \right) \quad (dx_1dy_1 = r_1 dr_1 d\theta_1) \\
& = |S_1^{2N-3}|  r 2\pi \int_0^{r/\sqrt{\xi}} r_1 (r^2-r_1^2)^\frac{2N-4}{2}~dr_1 \\
& = |S_1^{2N-3}|  r^{2N-1}2\pi \int_0^{1/\sqrt{\xi}} \mu (1-\mu^2)^{N-2}~d\mu \qquad (dr_1 = r d\mu)\\
& = |S_1^{2N-3}|   r^{2N-1} \pi  \int_0^{1/\xi} (1-t)^{N-2}~dt  \qquad (\mu^2 = t) \\
& = |S_1^{2N-3}|   r^{2N-1}\pi \frac{1}{N-1}\left[ 1- (1-\frac{1}{\xi})^{N-1} \right].
\end{aligned}
\]
In a similar manner, we can obtain the area of $S^{2N-1}_r$:
\[
|S^{2N-1}_r|  = |S_1^{2N-3}|   r^{2N-1} \frac{\pi}{N-1},
\]
Hence, by taking $r=1$ and $\xi = N$, we find:
\begin{align*}
\mathrm{Pr} \left(|b_1| \le \frac{1}{\sqrt{N}}\right) &= 1 - \left( 1-\frac{1}{N} \right)^{N-1} \mathop{\longrightarrow}^{N \rightarrow \infty} 1-e^{-1}\\
\mathrm{Pr} \left(|b_1| \ge \frac{1}{\sqrt{N}}\right) &=  \left( 1-\frac{1}{N} \right)^{N-1} \mathop{\longrightarrow}^{N \rightarrow \infty} e^{-1}.
\end{align*}

\end{proof}


The above result suggests that a randomly chosen initial state can have an overlap larger than $\frac{1}{N}$ with the unknown target state $\ket t$, with a probability of $1/e$. It implies that if we take a sufficiently large set of randomly-chosen initial states of size $K$, then we can find the minimum value $K=11$ such that $1-(1-1/e)^K\ge 0.99$. In other words, for a set of randomly-chosen initial states $\{\ket{y_k}\}$ with set size $11$, at least one of $\ket{y_k}$'s will have an overlap larger than $ \frac{1}{N}$ with $\ket t$, with probability larger than $0.99$. Such minimum value $K=11$ is independent of the size $N$ of the eigenvalue problem.

In spite of its usefulness in theory, a randomly chosen initial state from the uniform distribution is hard to prepare in practice. Alternatively, we aim to find a set of initial states which not only have the above nice property, but are also easy to generate in experiment. Fortunately, such initial states do exist. In the following, we can prove that the $11$ initial states can be simply chosen from the computational basis, predetermined by the quantum system. (Normally, we assume that the computational basis states are easy to prepare. For example, we assume the ground state $\ket{00\cdots 0}$ is easy to prepare, and all the other computational basis state can be prepared from the ground state with no more than $\log N$ rotations.)

Specifically, let $\{\ket{\phi_i}\}_{i=1}^m$ be a set of vectors chosen from the computational basis of $\CC^N$, ($m \ll N$), and let $\ket t$ be the target eigenstate. Since $\ket t$ is unknown before we solve the eigenvalue problem, we can assume it is randomly chosen from the uniform distribution on the unit sphere of $\CC^N$. Then we would like to evaluate the following probability:
\[
p_1=\text{Pr}\big(\max_{1\le i \le m} |\bra {\phi_i}t\rangle | \ge 1/\sqrt{\xi}\big).
\]
Without loss of generality, we can further assume $\{\ket{\phi_i}\}_{i=1}^m$ are the first $m$ number of computational basis vectors:
\[
\ket{\phi_i} = \ket{e_i} = (0,\cdots,0,\underbrace{1}_{i\text{th}},0 \cdots,0)^T.
\]
Since $\ket t$ can be considered as a complex vector $\bm{z}\in\CC^N$, we denote $\ket t=\bm{z}=(z_1,z_2,\ldots,z_N)^T$, where $z_j = x_j + y_j i$ with $x_j, y_j \in \mathbb{R}$. Then $\bm{z}$ can be considered as a vector in $\RR^{2N}$ with $\bm{z} = (x_1,y_1,\cdots,x_N,y_N)^T$. Let $S^{2N-1}$ denote the $(2N-1)$-dimensional unit sphere in $\RR^{2N}$,
and define
\[
\tilde{S}^{2N-1}\equiv \{\bm{z} \in \RR^{2N} : |\bra {\phi_i}t\rangle|^2 = x^2_i + y^2_i < 1/\xi, \quad i = 1,2, \cdots,m\} \subset S^{2N-1} \quad \quad (\xi \ge m).
\]
Then we have:
\[
p_1 = 1- \frac{\tilde{S}^{2N-1}}{S^{2N-1}}.
\]
It suffices to calculate the area of  $\tilde{S}^{2N-1}$. Noting that $|\bm{z}|^2<r^2$ is equivalent to $(x_1^2+y_1^2)+ \cdots + (x_N^2+y_N^2)<r^2$, and in use of the transform $dx_1dy_1 = r_1dr_1d\theta_1$
we calculate as
\[
\begin{aligned}
& \quad |\tilde{S}^{2N-1}| = \int_{S^{2N-1}} 1_{\tilde{S}^{2N-1}} ~dS^{2N-1} \\
& \mathop{=}^{r=1} \frac{\partial}{\partial r} \int_{|\bm{z}|<r} 1_{\{\frac{|z_i|}{|z|} < \frac{1}{\sqrt{\xi}}, 1 \le i \le m \}} (x_1,y_1,\cdots, x_N,y_N )~dx_1dy_1\cdots dx_N dy_N \\
& \mathop{=}^{r=1} \prod_{i=1}^m \left( \int_0^{\frac{r}{\sqrt{\xi}}} r_idr_i \int_0^{2\pi}d\theta_i \right)  \frac{\partial}{\partial r} \int_0^{\sqrt{r^2 - r_1^2 - \cdots -r_m^2}} \rho^{2N-2m -1}d\rho \int_{ S^{2N-2m-1}}~dS^{2N-2m-1} \\
& \mathop{=}^{r=1} r (2\pi)^m |S^{2N-2m-1}| \underbrace{ \int_0^{\frac{r}{\sqrt{\xi}}}\cdots \int_0^{\frac{r}{\sqrt{\xi}}} }_{m} r_1\cdots r_m (r^2 - r_1^2 - \cdots - r_m^2)^{N-m-1}~dr_1\cdots dr_m \\
& =  \pi^m |S^{2N-2m-1}|  \underbrace{ \int_0^{\frac{1}{\xi}}\cdots \int_0^{\frac{1}{\xi}} }_{m} (1-\mu_1 - \cdots - \mu_m)^{N-m-1}  \qquad \quad (\mu_i=r_i^2) \\
& = \frac{\pi^m |S^{2N-2m-1}|}{ (N-m)(N-m+1)\cdots(N-1) } \sum_{k=0}^m (-1)^k  \dbinom{m}{k}\left( 1-\frac{k}{\xi} \right)^{N-1}.
\end{aligned}
\]
In use of
\[
|S^{2N-1}| = \frac{\pi}{N-1} |S^{2N-2-1}| = \cdots = \frac{\pi^m}{(N-1)(N-2)\cdots (N-m)} |S^{2N-2m-1}|,
\]
we find that
\[
p_1 = 1- \frac{\tilde{S}^{2N-1}}{S^{2N-1}} = 1 - \sum_{k=0}^m (-1)^k  \dbinom{m}{k}\left( 1-\frac{k}{\xi} \right)^{N-1}.
\]
Taking $\xi = N$ and passing to the limit $N \rightarrow \infty$, we see that
\[
\begin{aligned}
p_1 = 1 - \sum_{k=0}^m (-1)^k  \dbinom{m}{k}\left( 1-\frac{k}{N} \right)^{N-1} \rightarrow 1- \sum_{k=0}^m (-1)^k  \dbinom{m}{k}e^{-k} = 1- (1-e^{-1})^m.
\end{aligned}
\]
Thus, the minimum value $m=11$ is sufficient to make $p_1\ge 0.99$. We summarize this result into the following:

\begin{theorem}

Let the unknown target state $\ket t$ be randomly chosen from the uniform distribution on the unit sphere of $\CC^N$.
If we prepare a set of $11$ initial states $\{\ket{\phi_k}\}$, $k=1,\cdots,11$, chosen from the computational basis of $\CC^N$, then the probability of at least one of $\ket{\phi_k}$'s satisfying $|\bra {\phi_k}t\rangle | \ge 1/\sqrt{N}$ is larger than $0.99$.

\end{theorem}

\end{document}